\documentclass[twoside,leqno,twocolumn]{article}

\usepackage[letterpaper]{geometry}

\usepackage{siamproceedings}

\usepackage[T1]{fontenc}
\usepackage{amsfonts}
\usepackage{graphicx}
\usepackage{epstopdf}
\usepackage{enumitem}
\usepackage{thm-restate}
\ifpdf
  \DeclareGraphicsExtensions{.eps,.pdf,.png,.jpg}
\else
  \DeclareGraphicsExtensions{.eps}
\fi

\newsiamremark{remark}{Remark}
\newsiamremark{hypothesis}{Hypothesis}
\crefname{hypothesis}{Hypothesis}{Hypotheses}
\newsiamthm{claim}{Claim}

\usepackage{amsopn}

\newcommand\relatedversion{}

\title{\Large Facility Location and $k$-Median with Fair Outliers\relatedversion}
   \author{Rajni Dabas\thanks{Northwestern University, Evanston IL 60208.}
    \and Samir Khuller\footnotemark[1]
    \and Emilie Rivkin\footnotemark[1]}

\date{}

\usepackage{amssymb,amsfonts,amsmath,amsopn,amstext,amscd,latexsym,xy,color,graphicx, verbatim, dsfont, enumitem, tikz, tikz-cd, xfrac,svg,hyperref}
\usepackage{algorithm,algpseudocode}
\usetikzlibrary{math,decorations.pathreplacing}
\usepackage{subcaption}

\usepackage[maxbibnames=99]{biblatex}
\newcommand{\Z}{{\mathbb{Z}}}

\newcommand{\cN}{{\mathcal{N}}}

\renewcommand{\phi}{\varphi}
\newcommand{\vare}{\varepsilon}

\newcommand{\lrp}[1]{\left({#1}\right)}

\renewcommand{\O}{\operatorname{OPT}}

\newcommand{\mss}{\operatorname{MSS}}
\newcommand{\kmss}{k\operatorname{MSS}}
\newcommand{\kmssg}{k\operatorname{MSS}_{\geq}}

\begin{document}
\onecolumn
\maketitle

 \begin{abstract}

Classical clustering problems such as \emph{Facility Location} and \emph{$k$-Median} aim to efficiently serve a set of clients from a subset of facilities---minimizing the total cost of facility openings and client assignments in Facility Location, and minimizing assignment (service) cost under a facility count constraint in $k$-Median. These problems are highly sensitive to outliers, and therefore researchers have studied variants that allow excluding a small number of clients as outliers to reduce cost. However, in many real-world settings, clients belong to different demographic or functional groups, and unconstrained outlier removal can disproportionately exclude certain groups, raising fairness concerns.

We study \emph{Facility Location with Fair Outliers}, where each group is allowed a specified number of outliers, and the objective is to minimize total cost while respecting group-wise fairness constraints. We present a bicriteria approximation with a $O(1/\epsilon)$ approximation factor and $(1+ 2\epsilon)$ factor violation in outliers per group. For \emph{$k$-Median with Fair Outliers}, we design a bicriteria approximation with a $4(1+\omega/\epsilon)$ approximation factor and $(\omega + \epsilon)$ violation in outliers per group improving on prior work by avoiding dependence on $k$ in outlier violations.  We also prove that the problems are W[1]-hard parameterized by $\omega$, assuming the Exponential Time Hypothesis.

We complement our algorithmic contributions with a detailed empirical analysis, demonstrating that fairness can be achieved with negligible increase in cost and that the integrality gap of the standard LP is small in practice.


\end{abstract}
 
\newpage
\twocolumn

\section{Introduction}


The {\em facility location} problem (FL) is a classical problem in combinatorial optimization in which we are given a collection of potential facility locations $\cal F$, each with an opening cost $f_i$. The goal is to serve a set of clients $C$ using the closest open facility with the objective function of minimizing the sum of the costs of opening facilities and the sum of the distances of all clients to their closest open facility. Once we select the facilities to open, the rest of the cost is completely determined, since each client simply connects to its closest open facility.

This classical problem has been extensively studied, with a variety of heuristics as well as fast approximation algorithms being developed - not just that - it has been the bedrock of showcasing different techniques for the problem, including greedy algorithms, LP-rounding algorithms, randomized rounding, Primal-Dual algorithms, etc \cite{Williamson_shmoys-book, JV_book}. So much so, that many valuable techniques can be traced back to early work on this problem. After many years of research, we finally have very fast and practical algorithms, even with an amazing understanding of their worst-case behavior.

Moreover, this problem is closely related to fundamental clustering problems, such as $k$-median ($k$M), which asks that $k$ cluster centers be chosen to minimize the sum of distances of all points to the closest cluster center. In fact, one of the first approximation algorithms for the $k$-median works by essentially reducing it to a collection of facility location instances with different facility opening costs with each point being viewed as a client. Multiple solutions are then combined by a randomized algorithm to produce a feasible solution \cite{jain_primal-dual_1999}. In most basic clustering problems, the user has to specify $k$, the number of clusters in advance without full knowledge of the data - with facility location by setting different values for facility costs we can obtain a spectrum of solutions that essentially creates clusters (nodes in $C$ that are connected to a common open facility).


All of these measures are highly sensitive to outliers in the data and can lead to very skewed solutions. This led researchers to consider the problem of clustering with outliers \cite{charikar_algorithms_2001}. Here, we wish to find the best possible clustering for a given  (user-specified) fraction  of the points. Several basic questions were considered - including $k$-centers, $k$-median and Facility Location. Over time, hundreds of papers have appeared that address the issue of outliers in many contexts \cite{charikar_algorithms_2001, jain_greedy_2003, chen_constant_2008, Friggstad_LS_2019, krishnaswamy_constant_2017, Gupta_kmo_2021, chakrabarty_non-uniform_2016,Harris_lottery_model_2019,Bhaskara_kmeans0_2019,Chen_neirips_outlier_detection_18}.

Our work is motivated by Almanza et. al. \cite{almanza_k-clustering_2022} who wrote "Clustering problems and clustering algorithms are often overly sensitive to the presence of outliers: even a handful of points can greatly affect the structure of the optimal solution and its cost. This is why many algorithms for robust clustering problems have been formulated in recent years. These algorithms discard some points as outliers, excluding them from the clustering. However, outlier selection can be unfair: some categories of input points may be disproportionately affected by the outlier removal algorithm."


Returning to the basic facility location problem, we are simply asked to serve a given number of clients with the open facilities (see Figure \ref{fig:fair}) so as to minimize the cost of opening facilities and serving a given number of clients in the cheapest possible way. Thus, we have no way to control which clients become outliers.

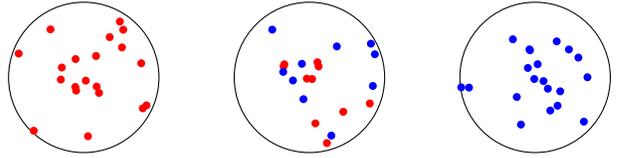
\begin{figure}[h]
\centering
\begin{tikzpicture}
\def\radius{1}

\def\xshift{3}

\draw (0,0) circle (\radius);
\foreach \i in {1,...,20} {
\pgfmathsetmacro{\angle}{rand*360}
\pgfmathsetmacro{\r}{rand*\radius}
\fill[red] ({\r*cos(\angle)}, {\r*sin(\angle)}) circle (1.5pt);
}

\draw (\xshift,0) circle (\radius);
\foreach \i in {1,...,10} {
\pgfmathsetmacro{\angle}{rand*360}
\pgfmathsetmacro{\r}{rand*\radius}
\fill[red] ({\xshift + \r*cos(\angle)}, {\r*sin(\angle)}) circle (1.5pt);
}
\foreach \i in {1,...,10} {
\pgfmathsetmacro{\angle}{rand*360}
\pgfmathsetmacro{\r}{rand*\radius}
\fill[blue] ({\xshift + \r*cos(\angle)}, {\r*sin(\angle)}) circle (1.5pt);
}

\draw (2*\xshift,0) circle (\radius);
\foreach \i in {1,...,20} {
\pgfmathsetmacro{\angle}{rand*360}
\pgfmathsetmacro{\r}{rand*\radius}
\fill[blue] ({2*\xshift + \r*cos(\angle)}, {\r*sin(\angle)}) circle (1.5pt);
}

\end{tikzpicture}
\caption{Illustration showing traditional clustering with outliers algorithms can disproportionately affect one group. We have two types of clients in three clusters of roughly the same size (20 clients in each cluster, 60 in all, 30 of each type). If we wish to provide service to a total of 20 clients at low cost, then we might pick the left cluster or the right cluster, but the center cluster is more "fair" in the sense that at least 10 clients from each community are served. With either the left cluster or the right cluster only members of one group are served. On the other hand if we wish to provide service to 40 clients, then the left and right clusters would be the right choice, but picking the middle and left cluster would serve one population very well, but the other very poorly.}
\label{fig:fair}
\end{figure}

In many applications, we might be dealing with clients with different attributes associated with them. For example, services in a city might need to be planned to serve different communities. If we utilize any of the algorithms with outliers, we have no control over which communities are unfairly dropped from consideration. This means that we need to provide better control over the outliers of different groups.



In particular, our model incorporates fairness across groups by explicitly controlling the number of outliers allowed in each community. Let $(\mathcal{M},d)$ denotes a metric space where $\mathcal{M}$ is a finite set of points and $d:\mathcal{M} \times \mathcal{M} \rightarrow \mathbb{R}^+$ is a distance function satisfying triangle inequality and symmetry. The model can now be formally defined as follows:

\begin{definition}[Facility Location with Fair Outliers] We are given a set $\mathcal{F} \subseteq \mathcal{M}$ of potential facility locations, and a set of clients $C \subseteq \mathcal{M}$ partitioned into $\omega$ disjoint groups: $C_1, C_2, \dots, C_{\omega}$. Each facility $i \in \mathcal{F}$ has an associated opening cost $f_i$. Additionally, for each group $C_g$, we are given an upper bound $\ell_g$ on the number of clients that may be designated as outliers. The goal is to select a subset of facilities $F \subseteq \mathcal{F}$ to open, and for each group $C_g$, a subset of outliers $C'_g \subseteq C_g$ with $|C'_g| \leq \ell_g$ to minimize the total cost:
\[
\sum_{i \in F} f_i + \sum_{g=1}^{\omega} \sum_{j \in C_g \setminus C'_g} d_{j,F}
\]
where $d_{j,F}$ denotes the cost of assigning client $j$ to its nearest open facility in $F$.
\end{definition}

The $k$-Median with Fair Outliers problem is same as Facility Location with Fair Outliers except that, instead of facility opening costs, the $k$-Median variant imposes a hard budget $k$ on the number of facilities that may be opened. Formally,

\begin{definition}[$k$-Median with Fair Outliers] We are given a set $\mathcal{F} \subseteq \mathcal{M}$ of potential facility locations, and a set of clients $C \subseteq \mathcal{M}$ partitioned into $\omega$ disjoint groups: $C_1, C_2, \dots, C_{\omega}$. For each group $C_g$, we are given an upper bound $\ell_g$ on the number of clients that may be designated as outliers. Additionally, we are given an integer $k$ bounding the number of facilities that may be opened. The goal is to select a subset of facilities $F \subseteq \mathcal{F}$ to open with $|F| \leq k$, and for each group $C_g$, a subset of outliers $C'_g \subseteq C_g$ with $|C'_g| \leq \ell_g$ to minimize :
\[
\sum_{g=1}^{\omega} \sum_{j \in C_g \setminus C'_g} d_{j,F}
\]
\end{definition}

These problems are $NP$-hard and so researchers have resorted to the design of heuristics and approximation algorithms.

\subsection{Our Contributions}

We first study facility location with fair outliers problem for arbitrary number of groups in Section \ref{sec:flfo} to give a $O(1/\epsilon)$ factor approximation for the problem with $(1+2\epsilon)$ factor violation in the outliers for any group. In particular, we present Theorem \ref{thm:flfo}.

Inamdar and Varadarajan~\cite{Inamdar_LP_SC_2018} gave an \(O(\log \omega)\)-approximation for the facility location problem with fair outliers for arbitrary \(\omega\). More recently, Bajpai et al.~\cite{bajpai2025flfo} studied the case of a constant number of groups and obtained a 4-approximation. However, both these results rely on solving an exponentially large linear program using the ellipsoid method, and are therefore impractical to implement even for small values of \(\omega >1\).

\begin{restatable}{theorem}{flfo}
\label{thm:flfo}
   There exists a polynomial time $O(1/\epsilon)$-factor approximation algorithm with $(1+2\epsilon)$ factor violation in outliers for each group for facility location with fair outliers problem where number of groups is arbitrary and $\epsilon>0$ is a fixed parameter.
\end{restatable}

We will show that the standard LP formulation for the problem has unbounded integrality gap even when $\omega=1$ and hence we present a  bi-criteria algorithm using LP rounding technique. Moreover, for arbitrary $\omega$, there is a known hardness of $\Omega(\log \omega)$ due to \cite{Inamdar_LP_SC_2018} unless $P = NP$ ruling out any proper (as opposed to bi-criteria)
polynomial-time approximation algorithm for an arbitrary number
of groups.

In this paper, we also give a stronger hardness result. In particular, we show via a series of reductions that the problem Facility Location with Outliers is {\em W[1]-hard} parameterized by \( \omega \), assuming the Exponential Time Hypothesis (ETH). A problem is W[1]-hard with respect to a parameter \( k \) if no algorithm can solve it in time \( f(k) \cdot \text{poly}(n) \), where \( f(k) \) is an arbitrary function of \( k \), and \( n \) is the size of the input, ETH is false. The ETH states that there is no algorithm that can solve 3-SAT (or any NP-complete problem) in sub-exponential time, i.e., in time \( 2^{o(n)} \) for \( n \) variables. In the context of our result, W[1]-hardness with respect to \( \omega \) implies that, no algorithm can solve the problem in time \( f(\omega) \cdot \text{poly}(n) \) for any function \( f(\omega) \), unless ETH fails.

\begin{theorem}
    Assuming ETH, facility Location with fair outliers problem is W[1]-hard parametrized by $\omega$.
\end{theorem}

We next study the $k$-median with fair outliers problem in Section~\ref{sec:kmfo}, where we present a bi-criteria approximation algorithm. To the best of our knowledge, the only prior result for this problem is the bi-criteria approximation by \cite{almanza_k-clustering_2022}. They gave a constant-factor approximation, however, their algorithm allows a violation of factor in $3k + 2$ outliers per group, which can be prohibitive in practice, especially when $k$ is large. In contrast, our algorithm eliminates this dependence on $k$ in the violation bounds, making it more practical and scalable. The result is formally stated in Theorem~\ref{thm:kmfo}.

\begin{theorem}
\label{thm:kmfo}
   There exists a polynomial time \(4\lrp{1+\omega/\epsilon}\) factor approximation algorithm with a \(\lrp{\omega+\epsilon}\) factor violation in outliers for each group for $k$-median with fair outliers problem where number of groups is arbitrary and $\epsilon>0$ is a fixed parameter. 
\end{theorem}

We also show that the hardness result for facility location extends to the $k$-median with fair outliers problem with slight modification in the reduction. In particular, we prove the following.

\begin{theorem}
Assuming ETH, $k$-median with fair outliers is W[1]-hard parameterized by $\omega$.
\end{theorem}

Finally, we implement our facility location and \(k\)-median algorithms and evaluate their performance on a real-world and synthetic datasets. For the facility location problem with fairness constraints, even for small values of \(\omega\) (e.g., 2 groups), existing approximation algorithms~\cite{Inamdar_LP_SC_2018, bajpai2025flfo}---though polynomial-time---are impractical due to their reliance on solving exponentially large linear programs via the ellipsoid method. In contrast, our bi-criteria algorithm is significantly more efficient in practice. Empirically, it consistently outperforms its worst-case guarantees in both cost and fairness. In fact, the cost is very close to the LP optimal value, and exhibits only negligible violations of the fairness constraints.

We also implement and evaluate a fast combinatorial greedy dual-fitting heuristic that entirely bypasses solving a linear program, which is the primary bottleneck in terms of running time for bi-criteria algorithm. This heuristic offers a significantly more scalable alternative, especially on large datasets, while still maintaining a reasonable performance guarantee. Empirically, it achieves solution costs within a factor of approximately 1.4 of the LP optimum and attains fairness comparable to the bi-criteria algorithm.

In summary, we observe that fairness our algorithms can ensure fairness with minimal increase in the cost of the solution with respect to non-fair algorithms (facility location with outliers). 

For the \(k\)-median problem, we implement the bi-criteria algorithm and compare its performance---in both cost and fairness---against the non-fair variant, namely \(k\)-median with outliers~\cite{charikar_algorithms_2001}. Similar to the facility location case, we find that (i) our algorithm consistently outperforms its worst-case guarantees in practice, and (ii) it achieves significantly better fairness, while incurring virtually no increase in cost compared to the unfair baseline. We do not directly compare with Almanza et al.~\cite{almanza_k-clustering_2022}, as their experiments focus on the \(k\)-means objective, whereas our focus is on \(k\)-median. Nonetheless, our empirical findings support their conclusions in the context of the \(k\)-median objective as well.

\subsection{Organization of the paper}
The rest of the paper is organized as follows. In Section \ref{sec:RW}, we review the related work. Sections \ref{sec:flfo} and \ref{sec:kmfo} give details of approximation algorithms for facility location with fair outliers and $k$-median with fair outliers respectively. In Section \ref{sec:exp} we layout our experimental results. The proof of W[1] hardness is discussed in Appendix \ref{sec:lb}. Our source code can be accessed on \textit{github}\footnote{anonymous.4open.science/r/FLO-final-8A20/}.

\section{Related Work}
\label{sec:RW}
Our work lies at the intersection of two important areas: the well-studied domain of clustering with outliers and fairness in clustering. We review relevant literature from both individual domains, as well as prior work at their intersection.

{\em \textbf{Clustering with Outliers:}} The concept of outliers in facility location problem was first introduced by Charikar et al. \cite{charikar_algorithms_2001}. They observed that the standard linear programming (LP) relaxation for the Facility Location with Outliers (FLO) problem has an unbounded integrality gap. To address this issue, Charikar et al. \cite{charikar_algorithms_2001} proposed a technique that involves guessing the cost of the most expensive facility in the optimal solution. After guessing the most expensive facility, they leverage the primal-dual framework of Jain and Vazirani \cite{jain_primal-dual_1999} to develop a 3-factor approximation algorithm. The approximation ratio was later improved to 2 by Jain et al. \cite{jain_greedy_2003} through a simple greedy algorithm, analysed using the dual-fitting technique. Charikar et al. \cite{charikar_algorithms_2001} also studied the outlier variant of $k$-Center and $k$-Median Problem. For the $k$-Center\footnote{$k$-Center is same as $k$-Median except in $k$-Center the goal is to minimize the maximum distance instead of total distance.} with Outliers problem, they presented a 3-factor approximation algorithm based on a greedy strategy, which was subsequently improved to a 2-factor approximation via LP rounding \cite{chakrabarty_non-uniform_2016,Harris_lottery_model_2019}. In the case of $k$-Median with Outliers ($k$MO), they proposed a bi-criteria approximation algorithm that achieves a cost approximation factor of $4(1+\frac{1}{\epsilon})$, while allowing a $(1+\epsilon)$ violation in the number of outliers. As with FLO, the standard LP relaxation for $k$MO also exhibits an unbounded integrality gap.  However, unlike FLO, it is not straightforward to overcome this challenge, making the $k$MO problem more challenging and less well-understood. The first constant-factor approximation for $k$MO was obtained by Chen \cite{chen_constant_2008} using a Lagrangian relaxation approach, inspired by the framework of Jain and Vazirani \cite{jain_primal-dual_1999}, combined with iterative local search. However, the approximation factor in Chen’s result, while constant, is relatively large and unspecified. Krishnaswamy et al. \cite{krishnaswamy_constant_2017} significantly improved upon this by applying iterative rounding on a strengthened LP formulation, yielding a $ 7.081$ approximation. This result was further refined by Gupta et al. \cite{Gupta_kmo_2021}, who improved the approximation factor to $(6.994 + \epsilon)$ through enhancements in the iterative rounding technique. Friggstad et al. \cite{Friggstad_LS_2019} employed natural multiswap local search heuristics to address outliers in the $k$-Median problem. Their approach provides a $(3+\epsilon)$-factor approximation with a $(1+\epsilon)$-factor violation in the cardinality. They also show that any constant size multiswap local search algorithm has unbounded locality gap for the problem, therefore, the violation in $k$ or number of outliers is inevitable in their algorithm.

Incorporating outliers into clustering problems such as FL and $k$M significantly increases the complexity of these problems. Techniques that perform well in the standard (non-outlier) setting — such as LP rounding, primal-dual methods, and local search — do not extend straightforwardly to their outlier variants, partially because;

\begin{enumerate} \item the standard linear programming (LP) relaxations for the outlier variants exhibit unbounded integrality gaps as shown in \Cref{sec:flfo} for facility location and by Charikar et. al.
\cite{charikar_algorithms_2001} for $k$-Median and, \item natural multiswap local search algorithms of constant size have unbounded locality gaps \cite{Friggstad_LS_2019}. \end{enumerate}

Consequently, the approximation guarantees in the presence of outliers are substantially worse and require more sophisticated algorithmic techniques. Currently, the best-known approximation ratios for FLO and $k$MO — are $2$ \cite{jain_greedy_2003} and $6.994 + \varepsilon$ \cite{Gupta_kmo_2021}, respectively. In comparison, their non-outlier counterparts — FL and $k$M — admit significantly better approximations of $1.488$ \cite{Li_FL_2011} and $2.613$ \cite{Gowda_km_2023}, respectively. Notably, the outlier variants do not currently have stronger lower bounds also; the best-known bounds remain at $1.463$ \cite{Guha_greedy_fl_1998} for FL and $1.763$ \cite{jain_greedy_lb_2002} for $k$M.

{\em \textbf{Fairness in Clustering:}} Over the past few years, there has been a surge of interest in incorporating fairness into clustering algorithms, largely driven by growing concerns around bias and equitable treatment in machine learning applications. This has led to the development of new algorithmic frameworks that enforce fairness constraints, such as demographic parity and group-level representation, ensuring that the clustering outcomes do not disproportionately disadvantage any particular group. A variety of fairness notions have been studied in this context, including balance constraints, proportional representation, and individual fairness. See survey by Chhabra et al. \cite{Chhabra-fairnessSruvey-2021} and references within.

{\em \textbf{Fairness for Outliers:}} In a parallel line of research, fairness has also been explored in clustering with outliers, where the goal is to ensure that the exclusion of certain data points as outliers does not unfairly impact individuals or communities. Traditional algorithms for clustering with outliers lack any control over which clients are discarded, which may inadvertently lead to biased decisions as shown empirically in \cite{almanza_k-clustering_2022} and in our experiments.

The idea of fairness in the presence of outliers was first introduced in the context of the vertex cover problem\footnote{Given a graph \(G=\lrp{V,E}\), find a minimum-weight subset \(U\subseteq V\) such that, for all \(e\in E\), at least one endpoint is in \(U\).} by Bera et al. \cite{bera_approximation_2014}. They called the problem as {\em partition vertex cover problem} and gave an $O(\log \omega)$ approximation for the problem which is best possible when $\omega$ is an arbitrary integer. Bandyapadhyay et al. \cite{Bandyapadhyay_lp_vc_2023} studied the problem when the vertices are unweighted. They achieve a $(2+\epsilon)$-approximation in time $n^{O(\omega/\epsilon)}$. Hong and Kao \cite{Hung_PartialVC_hypergraphs_2022} studied the problem in hypergraphs and presented a $(f \cdot H_\omega + H_\omega)$-approximation,
where $f$ is the maximum edge size and $H_\omega$ is the $\omega
^{th}$ harmonic number.
 
 In the context of facility location, Inamdar and Varadarajan \cite{Inamdar_LP_SC_2018} studied facility location with fair outliers. They showed that this problem is $(\log \omega)$-hard to approximate for an arbitrary $\omega$. They also presented a matching upper bound, $\mathcal{O}(\log \omega)$, approximation algorithm. Recently, Bajpai et al.~\cite{bajpai2025flfo} obtained a 4-approximation for the problem with constant number of groups. For the $k$-Center with fair outliers problem (Colorful $k$-Center), the first algorithm was proposed by Bandyapadhyay et al. \cite{bandyapadhyay_constant_2019}, providing a polynomial time 2-approximation while allowing $k+\omega$ centers. A subsequent result by Anegg et al. \cite{anegg_technique_2022} eliminated the violation in $k$ and obtained a 4-approximation in time $O(n^\omega)$, which was later improved to a 3-approximation in time $O(n^{\omega^2})$ by Jia et al. \cite{jia_fair_2022} .

 For $k$-Median with fair outliers, Almanza et al. \cite{almanza_k-clustering_2022} presented a bi-criteria approximation algorithm, that is, they presented an $O(1)$ approximation algorithm with $(3k+2)$ factor violation in outliers of every group. 
 
We continue this line of research for facility location and $k$-Median problem. As noted by Almanza et al. \cite{almanza_k-clustering_2022}, this setting is "quite flexible and allows one to enforce popular fairness constraints such as demographic parity \cite{barocas-ml-book}, calibration within groups \cite{Pleiss-fairness-calibration-2017}, statistical parity \cite{Dwork_awareness_2018}, diversity rules (e.g., 80 percent rule) \cite{biddle_adverse_2016}, and proportional representation rules \cite{Monroe-rep-1995}."

\section{Facility Location with Fair Outliers}
\label{sec:flfo}
In this section, we present a bi-criteria approximation algorithm for the Facility Location with Fair Outliers problem. We begin by formulating the problem as an integer linear program.

In this formulation, $y_i$ indicates whether facility $i$ is open, $z_j$ indicates whether client $j$ is an outlier, and $x_{ij}$ denotes whether client $j$ is served by facility $i$. Constraints \eqref{LPFLP_const1} and \eqref{LPFLP_const2} ensure that each client is either assigned to an open facility or designated as an outlier. Constraints \eqref{LPFLP_const3} impose bounds on the total number of outliers for each group.

\begin{align*}
\text{Minimize}\quad & \sum_{i \in \mathcal{F}} f_i y_i + \sum_{j \in C} \sum_{i \in \mathcal{F}} d_{ij} x_{ij} \\
\text{subject to} \quad \\
& \sum_{i \in \mathcal{F}} x_{ij} + z_j \geq 1 && \forall j \in C \tag{1} \label{LPFLP_const1} \\
& x_{ij} \leq y_i && \forall j \in C, i \in \mathcal{F} \tag{2} \label{LPFLP_const2} \\
& \sum_{j \in C_g} z_j \leq \ell_g && \forall g \in [\omega] \tag{3} \label{LPFLP_const3} \\
& x_{ij}, y_i, z_j \in \{0,1\} && \forall i \in \mathcal{F}, j \in C
\end{align*}

We relax the integer constraints, allowing $x_{ij}, y_i, z_j$ to take values in the continuous range $[0, 1]$, resulting in the LP relaxation. 

The LP formulation exhibits an unbounded integrality gap even when \(\omega = 1\). Consider the following example: suppose there is a single facility with opening cost \(f\), and \(M\) clients co-located at that facility. If the number of allowed outliers is \(M-1\), the LP can fractionally open the facility to an extent of \(1/M\) and serve each client to the same extent, effectively serving one full client in total. In contrast, any integral solution would need to fully open the facility and serve one client, incurring a cost of \(f\), which becomes arbitrarily large relative to the LP cost as \(M\) increases—thus leading to an unbounded integrality gap.
The example can be modified to allow a smaller number of outliers relative to the total number of clients by adding additional groups of clients served by facilities with zero opening cost. For these added groups, both the LP and the integral solutions coincide, so they do not affect the integrality gap, which still arises from the original group.

 Let $\rho^* = \langle x^*, y^*, z^* \rangle$ be an LP optimal solution for the LP. For any solution $\rho = \langle x, y, z \rangle$ to the LP, let $\text{cost}(\rho)$ denote its cost.

\subsection{Identifying the Outliers}

We now identify the set of clients to be treated as outliers in our solution. The idea is to declare a client as an outlier if it is predominantly an outlier in the LP solution $\rho^*$. For a given $\epsilon > 0$, we partition the client set $C$ into:
\begin{enumerate}
    \item[$(i)$] $C_o = \{ j \in C : z^*_j \geq 1 - \epsilon \}$
    \item[$(ii)$] $C_r = C \setminus C_o$
\end{enumerate}

We define a new solution $\hat{\rho} = \langle \hat{x}, \hat{y}, \hat{z} \rangle$ as follows:
\begin{enumerate}
    \item[$(i)$] For $j \in C_o$, set $\hat{z}_j = 1$ and $\hat{x}_{ij} = 0$ for all $i \in \mathcal{F}$.
    \item[$(ii)$] For $j \in C_r$, set $\hat{z}_j = z^*_j$ and $\hat{x}_{ij} = x^*_{ij}$ for all $i \in \mathcal{F}$.
    \item[$(iii)$] For all $i \in \mathcal{F}$, set $\hat{y}_i = y^*_i$.
\end{enumerate}

For any group $g \in [\omega]$, we observe:
\[
\sum_{j \in C_g} \hat{z}_j \leq \frac{1}{1 - \epsilon} \sum_{j \in C_g} z^*_j \leq (1 + 2\epsilon) \ell_g
\]
Moreover, $\text{cost}(\hat{\rho}) \leq \text{cost}(\rho^*)$.

\subsection{Reduction to Facility Location}

We now scale the assignment variables for $j \in C_r$ so that each such client is fully served, while maintaining feasibility. Let $\rho' = \langle x', y', z' \rangle$ be the updated solution:
\begin{enumerate}
    \item For all $j \in C_r$ and all $i \in \mathcal{F}$, set:
    \[
    x'_{ij} = \frac{\hat{x}_{ij}}{\sum_{i \in \mathcal{F}} \hat{x}_{ij}}.
    \]
    \item For each $i \in \mathcal{F}$, set:
    \[
    y'_i = \min \left\{1, \hat{y}_i \cdot \max_{j \in C_r : \hat{x}_{ij} > 0} \left\{ \frac{x'_{ij}}{\hat{x}_{ij}} \right\} \right\}.
    \]
\end{enumerate}

Note that for $j \in C_r$, $\sum_{i \in \mathcal{F}} \hat{x}_{ij} \geq \epsilon$, so:
\[
 x'_{ij} \leq \frac{\hat{x}_{ij}}{\epsilon}, \quad \text{and} \quad \hat{y}_i \leq y'_i \leq \frac{\hat{y}_i}{\epsilon}.
\]

It follows that $\rho'$ is a fractional feasible solution to the standard Facility Location problem with client set $C_r$, and:
\[
\text{cost}(\rho') \leq \frac{\text{cost}(\hat{\rho})}{\epsilon} \leq \frac{\text{cost}(\rho^*)}{\epsilon}.
\]

Let \(\langle \bar{x}, \bar{y} \rangle\) be a solution to the Facility Location problem on \(C_r\) with cost at most \(\alpha \cdot \mathrm{LP}_{C_r}\), where \(\mathrm{LP}_{C_r}\) is the cost of the optimal solution to the LP relaxation restricted to \(C_r\). Then, the solution \(\langle \bar{x}, \bar{y}, \hat{z} \rangle\) constitutes a \(\frac{\alpha}{\epsilon}\)-approximate solution to the Facility Location with Fair Outliers problem, violating the group outlier bounds by at most a factor of \((1 + 2\epsilon)\). Hence, we obtain the following theorem.

\begin{theorem}
    There exists a polynomial time $(\alpha/\epsilon)$-factor approximation algorithm with $(1+2\epsilon)$ factor violation in outliers for each group for facility location with fair outliers problem where $\alpha$ is the approximation factor for classical facility location problem, number of groups is arbitrary and $\epsilon>0$ is a fixed parameter.
\end{theorem}

Note that once the set of outliers has been identified, any standard facility location algorithm~\cite{jain_greedy_2003, arya_local_2004, ChudakS03} can be applied to determine which facilities to open for the remaining clients. In our experiments, after removing the outliers, we use a simple heuristic for facility location in which we directly round the fractional facility openings and assign remaining clients to nearest facility. The variable values are in fact close to 0/1. This is the primary reason we do not violate the outlier or cost constraints once we solve the LP.

\section{$k$-Median with Fair Outliers}
\label{sec:kmfo}

In this section, we focus on the \(k\)-Median with Fair Outliers problem. We develop a bi-criteria approximation algorithm leveraging the well-studied framework of \(k\)-Median with Penalties. The core idea is to encode fairness constraints via appropriately scaled penalties and then apply a known constant-factor approximation algorithm for the penalty-based variant. This approach allows us to recover both cost and fairness guarantees in the original problem.

We begin by defining the $k$-Median with Penalties problem.

\begin{definition}[\(k\)-Median with Penalties]
Given a metric space \((\mathcal{M}, d)\), a set of facilities \(\mathcal{F} \subseteq \mathcal{M}\), a set of clients \(C \subseteq \mathcal{M}\), penalties \(p_j \geq 0\) for each client \(j \in C\), and an integer \(k\) bounding the number of open facilities, the goal is to select (i) a subset \(F \subseteq \mathcal{F}\) with \(|F| \leq k\), and (ii) a subset \(C' \subseteq C\) of outliers paying penalties, minimizing the total cost
\[
\sum_{j \in C \setminus C'} d_{j,F} + \sum_{j \in C'} p_j,
\]
where \(d_{j,F}\) is the distance from client \(j\) to its nearest open facility in \(F\).
\end{definition}

Charikar et al.~\cite{charikar_algorithms_2001} provide a 4-approximation algorithm for this problem:

\begin{theorem}[\cite{charikar_algorithms_2001}]\label{thm:k-med-pen}
Let \(C\) be the assignment cost of the returned solution and \(P\) be the total penalty paid by outliers. If \(\O_P\) denotes the cost of the optimal solution to the penalty-based problem, then
\[
C + 4P \leq 4 \O_P.
\]
\end{theorem}

\subsection{Our Algorithm}

We design our approximation for \(k\)-Median with Fair Outliers as follows:

\begin{enumerate}
    \item Guess the optimal cost \(\O_C\) of the \(k\)-Median with Fair Outliers instance within a factor \((1+\varepsilon)\).
    \item For each group \(g \in [\omega]\) and each client \(j \in C_g\), set the penalty
    \[
    p_j = \frac{\O_C}{\gamma \cdot \ell_g},
    \]
    where \(\gamma > 0\) is a parameter to be chosen.
    \item Solve the corresponding \(k\)-Median with Penalties instance with penalties \(\{p_j\}\) using the 4-approximation algorithm from Theorem \ref{thm:k-med-pen}.
    \item Let \(F\) be the set of opened facilities, \(C_o\) be the set of outliers paying penalties, and \(\sigma\) be the assignment of remaining clients \(C \setminus C_o\), then output \(F\) as the set of opened facilities, \(C_o\) as the set of outliers, and \(\sigma\) as the assignment of remaining clients \(C \setminus C_o\) for the $k$-Median with Fair Outliers instance.
\end{enumerate}

\subsection{Analysis}

\begin{proposition}\label{prop:k-med-guess}
For colorful \(k\)-median with \(\omega\) groups and \(\ell_g\) outliers per group, the optimal cost \(\O_C\) can be guessed within a factor \((1+\varepsilon)\) using polynomially many trials.
\end{proposition}

\begin{proof}
Let \(d_{\min}\) and \(d_{\max}\) denote the smallest and largest pairwise distances in the metric, and \(\ell = \sum_{g=1}^\omega \ell_g\). Since each non-outlier client contributes at least \(d_{\min}\) and at most \(d_{\max}\) to the cost, we have
\[
(n - \ell) d_{\min} \leq \O_C \leq (n - \ell) d_{\max}.
\]
By considering powers of \((1+\varepsilon)\) within this range, we require only \(O\big(\log_{1+\varepsilon} n\big)\) guesses to identify an estimate within a \((1+\varepsilon)\) factor.
\end{proof}

\begin{theorem}
For the \(k\)-Median with Fair Outliers problem with \(\omega\) groups, there exists a polynomial-time algorithm that returns a solution with cost at most
\[
4\left(1 + \frac{\omega}{\gamma}\right) \O_C,
\]
and with at most \((\omega + \gamma) \ell_g\) outliers in each group \(g\).
\end{theorem}

\begin{proof}
Consider the optimal fair solution with cost \(\O_C\) and \(\ell_g\) outliers per group \(g\). Using the penalty assignment defined above, this solution is feasible for the \(k\)-Median with Penalties problem so,
\[
\O_P \leq \O_C+\sum_{g=1}^\omega\frac{\O_C}{\gamma\ell_g}\ell_g = \O_C\lrp{1+\frac{\omega}{\gamma}}.
\]

We can also take this instance of \(k\)-medians with penalties and run it through the black box algorithm. It returns a solution of cost \(C\) with \(\ell'\) outliers. We partition the outliers by color. By Theorem \ref{thm:k-med-pen} and the above inequality, 
\[
C+4\sum_{c=1}^\omega\frac{\O_C}{\gamma\ell_c}\ell'_g \leq 4\O_P \leq 4\O_C\lrp{1+\frac{\omega}{\gamma}}.
\]
Since all terms are non-negative, we have 
\[
C \leq 4\O_C\lrp{1+\frac{\omega}{\gamma}}
\]
and 
\[
\sum_{c=1}^\omega\frac{1}{\gamma\ell_g}\ell'_g \leq 1+\frac{\omega}{\gamma}.
\]
Again, each term is non-negative, so
\[
\ell'_g \leq \ell_g\lrp{\omega+\gamma}.
\]

By Proposition \ref{prop:k-med-guess}, we can guess \(\O_C\) within a \(\lrp{1+\vare}\) factor, so the guarantees on \(C\) and \(\ell'_g\) for each group \(g\) hold within the same factor.
\end{proof}

\section{Experiments}
\label{sec:exp}
Our source code can be accessed on \textit{github}\footnote{anonymous.4open.science/r/FLO-final-8A20/}. The algorithms are implemented in Julia and run on a ThinkPad X1 laptop. In all our experiments, the input points are in the
Euclidean space, and we use the $\ell_2$ distance function. 

\textbf{Datasets.}
We use publicly available datasets, as preprocessed in \cite{almanza_k-clustering_2022}. The \textit{Adult} dataset~\cite{uci_adult} contains U.S. census information, where we use the \texttt{sex} attribute as the group label. The \textit{Bank} dataset~\cite{uci_bank} comprises data from a direct marketing campaign by a bank, with group labels based on the \texttt{marital status} attribute. For both datasets, we retain only the available numeric attributes, each of which is normalized independently.

For the facility location experiments, we randomly sample 4,500 data points and select 100 potential facilities from each dataset. For $k$-median, due to its computational overhead on a non-commercial machine, we only use a smaller subset of 500 points and set $k = 5$.

Note that both datasets are fully unsupervised — that is, they lack ground-truth cluster labels. Consequently, we cannot directly evaluate clustering accuracy. To address this limitation, we also evaluate our algorithms on synthetic datasets with known cluster structure. The \textit{Synthetic} dataset includes two group designations: ``in'' and ``out''. In-group members are drawn from the distribution $\cN(0, 10)$, while out-group members are drawn from $\cN(10, 20)$. Facilities are distributed according to the in-group distribution. A facility located within 10 units of the origin incurs a cost of 80; otherwise, its cost is 40. We generate 500 ``in'' and 50 ``out'' points for both the facility location and $k$-median experiments. 


\subsection{Facility Location}
\begin{figure}[t]
    \centering
    \begin{subfigure}[b]{0.8\linewidth}
        \centering
        \includegraphics[width=\linewidth]{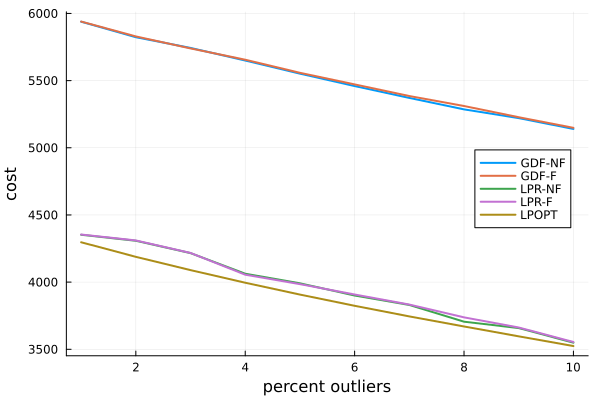}
        \caption{Cost vs.\ percentage of outliers.}
        \label{fig:FL-cost}
    \end{subfigure}
    
    \vspace{1em} 
    
    \begin{subfigure}[b]{0.8\linewidth}
        \centering
        \includegraphics[width=\linewidth]{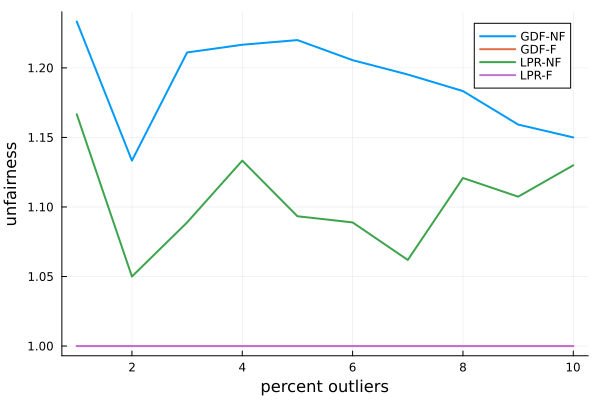}
        \caption{Unfairness vs.\ percentage. The lines for \textbf{GDF-F} and \textbf{LPR-F} overlap at 1.0.}
        \label{fig:FL-unf}
    \end{subfigure}
    
    \caption{Performance of different algorithms on the \textit{adult} dataset grouped by \textit{sex} for the facility location problem. We use the dataset with \(\varepsilon=0.1\), \(n=4500\), \(m=100\), and facility cost set uniformly at \(d_{\max}=18.32\).}
    \label{fig:FL-combined}
\end{figure}

\begin{figure}[t]
    \centering
    \begin{subfigure}[b]{0.8\linewidth}
        \centering
        \includegraphics[width=\linewidth]{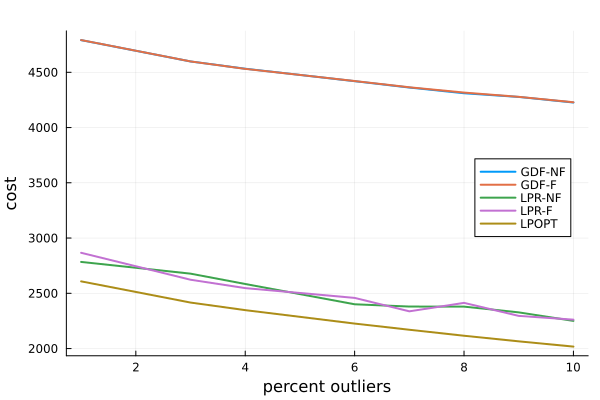}
        \caption{Cost vs.\ percentage of outliers.}
        \label{fig:FLbank-cost}
    \end{subfigure}
    
    \vspace{1em}
    
    \begin{subfigure}[b]{0.8\linewidth}
        \centering
        \includegraphics[width=\linewidth]{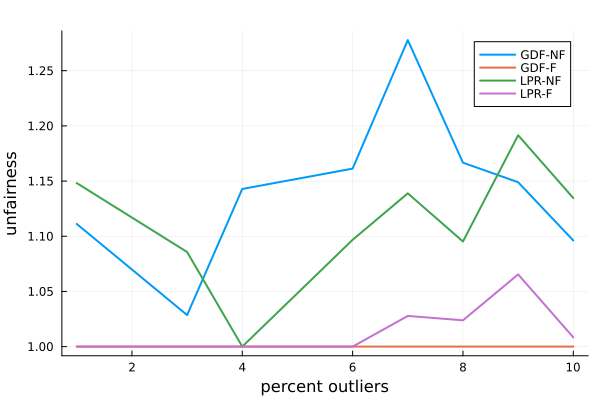}
        \caption{Unfairness vs.\ percentage of outliers.}
        \label{fig:FL-unf-bank}
    \end{subfigure}
    
    \caption{Performance of different algorithms on the \textit{bank} dataset grouped by \textit{marital status} for the facility location problem. We use the dataset with \(\varepsilon=0.5\), \(n=4520\), \(m=100\), and facility cost set uniformly at \(d_{\max}=24.57\).}
    \label{fig:FL-bank-combined}
\end{figure}

\begin{figure}[t]
    \centering

    \begin{subfigure}[b]{0.8\linewidth}
        \centering
        \includegraphics[width=\linewidth]{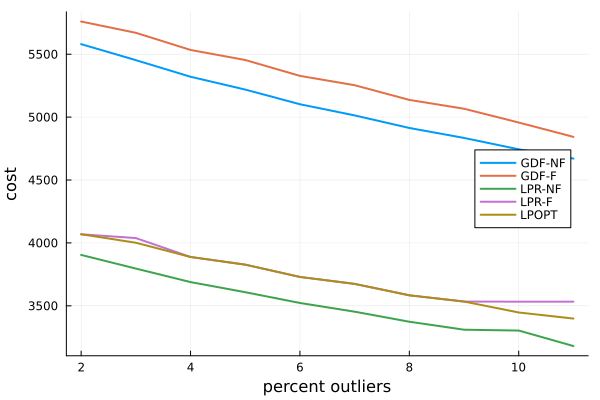}
        \caption{Cost vs.\ percentage of outliers.}
        \label{fig:FL-costSynth}
    \end{subfigure}
    
    \vspace{1em}
    
    \begin{subfigure}[b]{0.8\linewidth}
        \centering
        \includegraphics[width=\linewidth]{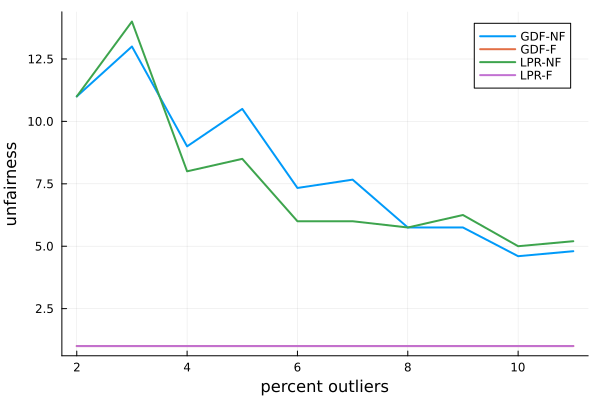}
        \caption{Unfairness vs.\ percentage of outliers. The lines for \textbf{GDF-F} and \textbf{LPR-F} overlap near \(1.0\).}
        \label{fig:FL-unfSynth}
    \end{subfigure}

    \caption{Performance of different algorithms on the \textit{synthetic} dataset with known group structure for the facility location problem. We use \(\varepsilon = 0.1\).}
    \label{fig:FL-synth-combined}
\end{figure}

We evaluate two algorithms for the Facility Location with Fair Outliers problem---our \emph{bi-criteria algorithm} based on LP rounding, and a \emph{greedy heuristic} that modifies the classical dual fitting approach to enforce group-wise fairness (does not involve solving an LP and is more scalable). We compare the cost and fairness of these algorithms against classical facility location with outliers variants of the two algorithms that do not enforce group-level fairness constraints. Additionally, we compare the solution costs of all algorithms to the LP optimal value for benchmarking.

To generate candidate facility locations, we apply the $k$-means algorithm of Khandelwal et al.~\cite{khandelwal_faster_2017} to obtain 100 potential facilities. Since clients rarely connect to distant facilities, we restrict the set of assignment variables \(x_{ij}\) to pairs with distances below the median of all client--facility distances. These preprocessing steps significantly improve the runtime of all algorithms. We assume uniform facility opening costs.

We next describe the four algorithms in detail:
\begin{itemize}
    \item \textbf{LPR-F (LP Rounding with Fairness)}: This is our main bi-criteria algorithm. It begins by solving the fair LP relaxation with group-level outlier constraints and declares clients outliers if they are outliers to a large extent in the LP solution. In the second phase, it applies a simple heuristic for facility location problem: facilities with fractional opening values above a certain threshold are opened, and each non-outlier client is assigned to its nearest open facility.
    
    \item \textbf{LPR-NF (LP Rounding, Non-Fair)}: A baseline variant of LPO+R that does not enforce fairness. The LP includes only a single constraint on the total number of outliers. The rounding process follows as in LPO+R.

    \item \textbf{GDF-F (Greedy Dual Fitting with Fairness)}: This algorithm modifies the classical greedy dual fitting method of Jain et al.~\cite{jain_greedy_2003} to enforce fairness in outlier selection. As before, each client $j$ maintains a variable $\alpha_j$, which increases uniformly over time. Facilities track surplus from clients and are opened when the accumulated surplus equals the opening cost. Once open, a facility's cost is reset to zero. We run the process until the required number of clients are connected in each group. If, during an iteration, the coverage requirement for any group is exceeded, we remove the remaining clients from that group in subsequent iterations. Note that, \textit{this algorithm does not require solving a linear program and hence is more scabale.}

    \item \textbf{GDF-NF (Greedy Dual Fitting, Non-Fair)}: The classical greedy facility location with outliers algorithm~\cite{jain_greedy_2003} without fairness constraints. Clients are connected greedily using the dual fitting process until the overall outlier budget is met, with no group-wise consideration.

\end{itemize}


\paragraph{Fairness Metric.}
To evaluate fairness, we use the \emph{unfairness} metric defined as:
\[
\text{unfairness}\footnote{We define unfairness to be at least $1$ since leaving fewer than the allowed number of outliers for a group (i.e., $\ell'_g < \ell_g$) does not violate the fairness constraint.} = \max \left\{ 1, \max_g \frac{\ell'_g}{\ell_g} \right\},
\]
where $\ell_g$ is the number of outliers allowed for group $g$, and $\ell'_g$ is the number of outliers left by the algorithm for group $g$. Values close to 1 indicate a fair solution.

We measure the \emph{solution cost} as the sum of facility opening and client connection costs. To understand the behavior of the algorithms as the number of allowable outliers varies (expressed as a percentage of the total number of clients), we plot 
\begin{enumerate}
    \item the solution costs for different algorithms as a function of the outlier budget, and
    \item the "unfairness" of different algorithms as a function of the outlier budget.
\end{enumerate}

\subsubsection{Results and Insights}

\begin{itemize}
    \item Both fair algorithms (LPR-F and GDF-F) achieve \emph{unfairness values equal to 1 in almost all cases}, with the only exception being the \textit{bank} dataset, where the maximum unfairness reaches 1.12. In contrast, the non-fair baselines (LPR-NF and GDF-NF) exhibit significantly higher group-level unfairness: up to 1.67 and 1.23 for the \textit{adult} dataset, 1.19 and 1.28 for the \textit{bank} dataset, and as high as 14 and 13 for the \textit{synthetic} dataset (see Figures~\ref{fig:FL-unf},~\ref{fig:FL-unf-bank}, and~\ref{fig:FL-unfSynth}).
    \item The \emph{increase in cost} due to enforcing fairness is \emph{negligible} (see Figures \ref{fig:FL-cost}, \ref{fig:FLbank-cost} and \ref{fig:FL-costSynth}):
    \begin{itemize}
        \item \textbf{LPR} (with and without fairness) is consistently close to LP optimal.
        \item \textbf{GDF-F} has costs comparable to the classical greedy baseline.
    \end{itemize}
\end{itemize}

These results demonstrate that \emph{fairness can be achieved without a substantial increase in cost}. Moreover, our bi-criteria algorithm performs well beyond worst-case guarantees—achieving low cost and near-perfect fairness. The LP solution—after the outliers are removed—exhibits opening and assignment variables that are typically very close to 0 or 1. This near-integrality of the solution enables a simple rounding heuristic to perform remarkably well in practice. By leveraging the structure of the LP output, our heuristic can make effective discrete decisions with minimal loss in quality, contributing both to the algorithm’s speed and its empirical effectiveness.

\subsection{$k$-Median}

\begin{figure}[t]
    \centering

    \begin{subfigure}[b]{0.8\linewidth}
        \centering
        \includegraphics[width=\linewidth]{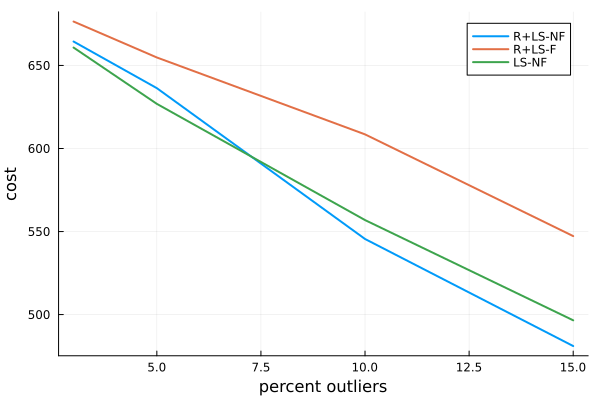}
        \caption{Cost vs.\ percentage of outliers.}
        \label{fig:k-costAdult}
    \end{subfigure}
    
    \vspace{1em}
    
    \begin{subfigure}[b]{0.7\linewidth}
        \centering
        \includegraphics[width=\linewidth]{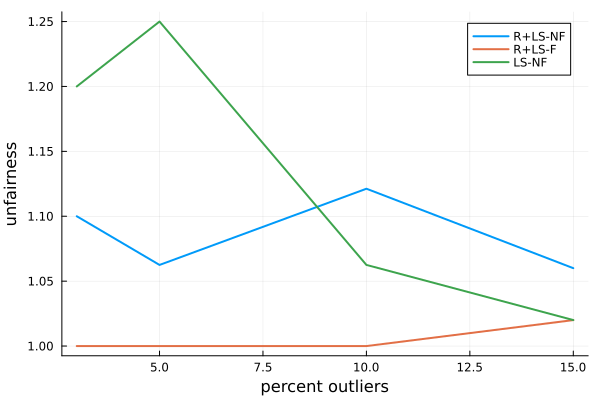}
        \caption{Unfairness vs.\ percentage of outliers.}
        \label{fig:k-unfAdult}
    \end{subfigure}

    \caption{Performance of different algorithms on the \textit{adult} dataset grouped by \textit{sex} for the $k$-median problem.}
    \label{fig:k-adult-combined}
\end{figure}

\begin{figure}[t]
    \centering

    \begin{subfigure}[b]{0.8\linewidth}
        \centering
        \includegraphics[width=\linewidth]{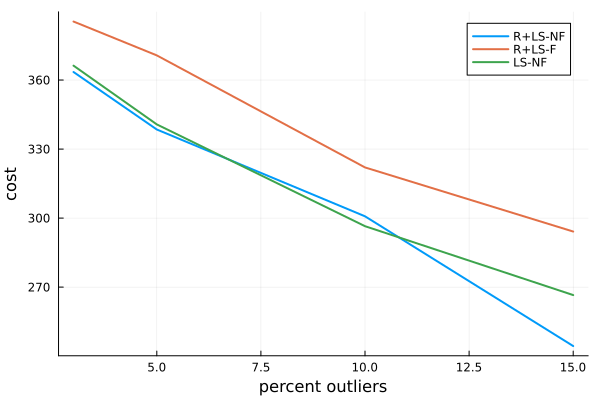}
        \caption{Cost vs.\ percentage of outliers.}
        \label{fig:k-costBank}
    \end{subfigure}
    
    \vspace{1em}
    
    \begin{subfigure}[b]{0.8\linewidth}
        \centering
        \includegraphics[width=\linewidth]{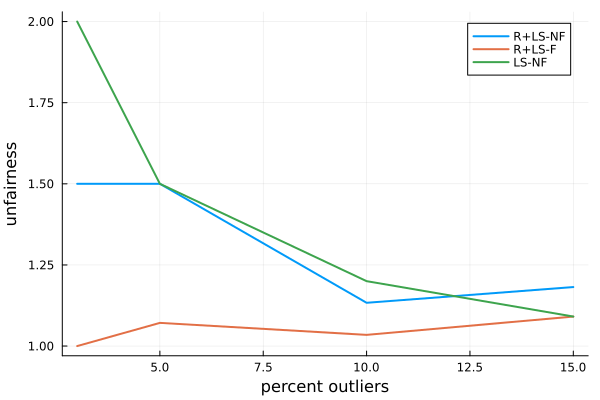}
        \caption{Unfairness vs.\ percentage of outliers.}
        \label{fig:k-unfBank}
    \end{subfigure}

    \caption{Performance of different algorithms on the \textit{bank} dataset grouped by \textit{marital status} for the $k$-median problem.}
    \label{fig:k-bank-combined}
\end{figure}

\begin{figure}[t]
    \centering

    \begin{subfigure}[b]{0.8\linewidth}
        \centering
        \includegraphics[width=\linewidth]{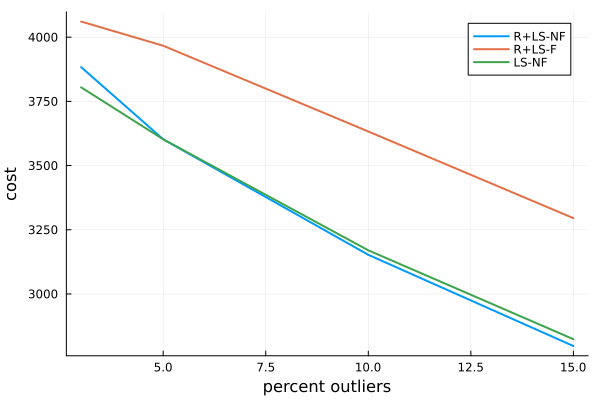}
        \caption{Cost vs.\ percentage of outliers.}
        \label{fig:k-costSynth}
    \end{subfigure}
    
    \vspace{1em}
    
    \begin{subfigure}[b]{0.8\linewidth}
        \centering
        \includegraphics[width=\linewidth]{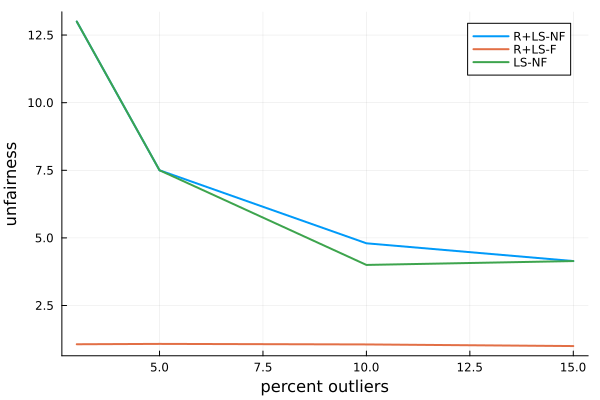}
        \caption{Unfairness vs.\ percentage of outliers.}
        \label{fig:k-unfSynth}
    \end{subfigure}

    \caption{Performance of different algorithms on the \textit{synthetic} dataset with known group structure for the $k$-median problem.}
    \label{fig:k-synth-combined}
\end{figure}

For $k$-Median with fair outliers, we implement the algorithm described in \Cref{sec:kmfo} with a minor modification and compare the fairness and cost of our algorithm against two baselines for $k$-Median with outliers (non-fair). The three algorithms we evaluate are:

\begin{itemize}
    \item \textbf{R+LS-F (Reduction + Local Search, Fair):} We implement the algorithm described in \Cref{sec:kmfo} with a minor modification: instead of using the primal-dual algorithm of Charikar et al.~\cite{charikar_algorithms_2001} for \(k\)-Median with penalties, we employ the local search algorithm proposed by Wang et. al.~\cite{wang_kmp}. The primal-dual algorithm is computationally expensive due to the use of Lagrangian relaxation; specifically, it requires solving two facility location subproblems—one using fewer than \(k\) facilities and one using more—through a binary search, with each step involving a call to a facility location solver. This significantly increases the runtime in practice. In contrast the local search is efficient. The algorithm begins with an arbitrary set of \(k\) open facilities and repeatedly considers swapping one currently open facility with one currently closed facility. Specifically, for each pair \((f_{\text{out}}, f_{\text{in}})\) where \(f_{\text{out}}\) is an open facility and \(f_{\text{in}}\) is closed, it computes the cost of the solution obtained by replacing \(f_{\text{out}}\) with \(f_{\text{in}}\), reassigning each client to its nearest open facility. If such a swap results in a decrease in total cost of at least 1\%, it is accepted. This process continues until no improving swap exists, at which point the algorithm terminates.

    Importantly, substituting the local search algorithm for the primal-dual method does not materially affect the worst-case approximation guarantees: the cost bound remains the same, and the violation bound increases only slightly.

    \item \textbf{R+LS-NF (Reduction + Local Search , Non Fair)} The algorithm is same as R+LS except we have no notion of group-wise fairness. The penalties are defined according to the total number of outliers. This is the algorithm for $k$-median with outlier given by Charikar et. al. \cite{charikar_algorithms_2001}.

    \item \textbf{LS-NF (Local Search, Non-Fair):} This algorithm performs standard local search for the \(k\)-Median objective, starting with an arbitrary set of \(k\) open facilities and repeatedly considering facility swaps as in R+LS. After convergence, the farthest \(\ell\) clients (based on their distances to the nearest open facility) are dropped as outliers.
\end{itemize}

We use the same \emph{unfairness} metric to measure the fairness of a solution and evaluate cost as the total client-to-facility connection cost. We plot the same charts as in the facility location experiments: cost vs.\ percentage of outliers and unfairness vs.\ percentage of outliers for all three algorithms.

\subsubsection{Results and Insights}

\begin{itemize}
    \item 
    The fair algorithm (R+LS-F) achieve \emph{unfairness close to 1}, while the non-fair baselines R+LS-NF and LS-NF exhibit significant group-level unfairness of up to 1.12 and 1.25 for adult dataset, 1.50 and 2 for bank dataset, and 13 and 13 for synthetic dataset, respectively (see Figures \ref{fig:k-unfAdult}, \ref{fig:k-unfBank} and \ref{fig:k-unfSynth}).
    \item enforcing fairness results in some increase in cost, it is not prohibitively large and is justified by the improved group-level guarantees (see Figures \ref{fig:k-costAdult}, \ref{fig:k-costBank} and \ref{fig:k-costSynth}).
\end{itemize}

As with the facility location experiments in this paper and the \(k\)-means results of Almanza et al.~\cite{almanza_k-clustering_2022}, our findings for \(k\)-Median demonstrate that \emph{fairness can be achieved with no significant increase in cost}.

\newpage
\printbibliography
\newpage
\appendix
\section{Lower Bounds}
\label{sec:lb}
\begin{theorem}
Assuming ETH, facility location with fair outliers is W[1]-hard parametrized by $\omega$.
\end{theorem}

We use the following lemmas to prove this statement. First, we need to define the problems we will use along the way.

\begin{definition}[Multidimensional Subset Sum \cite{inamdar_parameterized_2023}]
This is a known W[1]-hard problem. An instance \((V,T\) of \(\mss\) consists of \(n\) \(d\)-dimensional vectors with non-negative entries in \(\Z\) and a target vector \(T\). We want to determine if there is a subset \(U\subseteq V\) such that \(\sum_{v\in U}v=T\).
\end{definition}

\begin{definition}[\(\kmss\)]
This is a variant on \(\mss\) where we want to know if there is a suitable \(U\) with cardinality \(k\).
\end{definition}

\begin{definition}[\(\kmssg\)]
Given a set of \(d\)-dimensional vectors \(V\) with non-negative, integer-valued entries and a target vector \(T\), we want to find a subset \(U\subseteq V\) such that \(|U|\leq k\) and for all \(j=1,\ldots,d\), \(\sum_{v_j\in U}\geq t_j\).
\end{definition}

\begin{lemma}
\(\kmss\) is W[1]-hard.
\end{lemma}

\begin{proof}

Take an instance \(\lrp{V,t}\) of \(\mss\). For each \(k=1,\ldots,n\), use the decider for \(\kmss\) to determine if \(\lrp{V,t,k}\) is an instance of \(\kmss\). If \(\lrp{V,t}\in\mss\), then there exists some \(k\) for which \(\lrp{V,t,k}\in \kmss\). If \(\lrp{V,t}\notin\mss\), then there is no such \(k\). Therefore, we can use the decider for \(\kmss\) to solve an instance of \(\mss\).

\end{proof}

\begin{lemma}
\(\kmssg\) is W[1]-hard.
\end{lemma}

\begin{proof}

Take an instance \(\lrp{V,t,k}\) of \(\kmss\). For a sufficiently large constant \(L\) and for each \(k'=1,\ldots,k\), consider an instance \(\lrp{V',t',k'}\) of \(\kmssg\) where \(v'=\lrp{v_1,\ldots,v_d,L-v_1,\ldots,L-v_d}\) for every \(v\in V\) and \(t'=\lrp{t_1,\ldots,t_d,k'L-t_1,\ldots,k'L-t_d}\). 

We can use the decider for \(\kmssg\) to determine if \(\lrp{V,t,k}\in\kmss\). We assumed that \(L\) is sufficiently large (perhaps \(k\cdot t_1\cdots t_d\)). As a result, we do not have to worry about the case where \(k'L-t_i\leq(k'-1)L\); we can treat the sum to \(k'L\) and the sum to \(-t_i\) separately. 

If \(\lrp{V',t',k'}\in \kmssg\), then there are exactly \(k'\) vectors in the set \(U'\) with
\[
\sum_{v'\in U'}v'_j = \sum_{v'\in U'}L-v_j \geq t'_j = k'L-t_j
\]
for \(j=d+1,\ldots,2d\). Since we know that the \(k'L\) is only coming from the contribution of \(L\) from each vector, this implies that
\[
\sum_{v'\in U'}v_j \leq t_j.
\]
Moreover, those same vectors in \(U'\) have
\[
\sum_{v'\in U'}v'_j = \sum_{v'\in U'} v_j \geq t_j
\]
for \(j=1,\ldots,d\). If a set \(U'\) exists, then it is a size \(k'\leq k\) subset of \(V'\) with
\[
\sum_{v'\in U'}v_j \leq t_j.
\]
Therefore, \(\lrp{V,t,k}\in \kmss\). If there is no \(k'\leq k\) where \(\lrp{V',t',k'}\in\kmssg\), then \(\lrp{V,t,k}\notin \kmss\).

\end{proof}

Now, we can prove the theorem. 

\begin{proof}
For each vector, create a facility with co-located clients. These facilities are some large distance away from each other. Specifically, for each vector \(v\in V\), create a facility \(i\) with \(f_i=1\). For each dimension \(g=1,\ldots,d\), place \(v_g\) clients of group \(g\) colocated with the facility. Now, we have an instance of facility location with fair outliers, \(\lrp{F,C,\ell}\), where we have to cover at least \(\ell_g\) clients from each group. 

By using the decider for the facility location with fair outliers problem with cost no more than \(k\), we can solve the instance of \(\kmssg\).
\end{proof}
    
\begin{corollary}
Assuming ETH, $k$-median with fair outliers is W[1]-hard parametrized by $\omega$.
\end{corollary}
\end{document}